\documentclass[12pt]{article}
\usepackage{euscript,amsmath,amssymb,latexsym,amsthm,enumerate,mathrsfs}
\usepackage[bookmarks=true]{hyperref}
\usepackage{cite}
\usepackage{enumerate}
\usepackage{graphicx,xcolor}

\DeclareMathOperator{\supp}{supp}

\theoremstyle{definition}

\theoremstyle{remark}
\newtheorem{remark}{Remark}

\theoremstyle{plain}
\newtheorem{theorem}{Theorem}
\newtheorem{lemma}{Lemma}

\begin{document}
\title{
\vspace{1cm} {\bf Microscopic and soliton-like solutions of the Boltzmann--Enskog and generalized Enskog  equations for elastic and inelastic\\ hard spheres}
}
\author{A.\,S.~Trushechkin\bigskip
 \\
{\it  Steklov Mathematical Institute of the Russian Academy of Sciences}
\\ {\it Gubkina St. 8, 119991 Moscow, Russia}\medskip\\
{\it  National Research Nuclear University ``MEPhI''}
\\ {\it Kashirskoe Highway 31, 115409 Moscow, Russia}\bigskip
\\ e-mail:\:\href{mailto:trushechkin@mi.ras.ru}{\texttt{trushechkin@mi.ras.ru}}}

\date{}
\maketitle

\begin{abstract}

N.\,N.~Bogolyubov  discovered that the Boltzmann--Enskog kinetic equation has microscopic solutions. They have the form of sums of delta-functions and correspond to trajectories of individual hard spheres. But the rigorous sense of the product of the delta-functions in the collision integral was not discussed. Here we give a rigorous sense to these solutions by introduction of a special regularization of the delta-functions. The crucial observation is that the collision integral of the Boltzmann--Enskog equation coincides with that of the first equation of the BBGKY hierarchy for hard spheres if the special regularization to the delta-functions is applied. This allows to reduce the nonlinear Boltzmann--Enskog equation to the BBGKY hierarchy of linear equations in this particular case.

Also we show that similar functions are exact smooth solutions for the recently proposed generalized Enskog equation. They can be referred to as ``particle-like'' or ``soliton-like'' solutions and are analogues of multisoliton solutions of the Korteweg--de Vries equation. 
\end{abstract}

\section{Introduction}

Kinetics of hard sphere gases is a classical field of physical and mathematical research and it is still actively developed \cite{GerasGap, Simon, Gran, BorGer}. One of kinetic equations in the focus of attention is the Boltzmann--Enskog kinetic equation. It is a generalization of the Boltzmann equation for the case of moderately dense (rather than dilute) gases. As the Boltzmann equation, it describes the time-irreversible behaviour and entropy production \cite{BelLach,Resibois}.

However, N.\,N.~Bogolyubov \cite{Bogol75} (see also \cite{BogBog}) discovered that this equation has time-reversible microscopic solutions as well. They have form of sums of delta-functions and correspond to trajectories of individual hard spheres. Another interest feature of these solutions (along with their reversibility) is that the Boltzmann--Enskog
kinetic equation is believed to give only an approximate aggregate
description of the gas dynamics in terms of a single-particle
distribution function. The Bogolyubov's result shows that the
Boltzmann--Enskog equation can describe the gas dynamics in terms of
trajectories of individual hard spheres as well.

Also the Vlasov kinetic equation is known to have microscopic
solutions. They were discovered by A.\,A.~Vlasov himself
\cite{Vlasov}.

However, the justification of these solutions for the Boltzmann--Enskog equation is performed by Bogolyubov on the ``physical'' level of rigour. In particular, Bogolyubov did
not discuss the products of generalized functions in the collision
integral. 

A way to give a rigorous sense to the microscopic solutions is the use of a regularization for delta-functions. A rather general type of regularization and a direct proof of the existence of the limit for the collision integral is proposed in \cite{Trush}. But the proof is rather involved. 

Here we present a special kind of regularization. The crucial observation is that, under this regularization, the collision integral of the Boltzmann--Enskog equation coincides with that of the first equation of the BBGKY hierarchy for hard spheres. This allows to reduce the nonlinear Boltzmann--Enskog equation to the BBGKY hierarchy of linear equations in this particular case. This gives a rigorous sense to the microscopic solutions in a simpler and, in some sense, more natural and elegant way.

Also we show that similar construction leads to exact smooth solutions for the recently proposed generalized Enskog equation \cite{GerasGap}. In general, this equation describes the dynamics of infinite number of hard spheres. But our solutions are expressed in terms of the evolution operator of a  finite number of hard spheres.  These solutions can be referred to as ``particle-like'' (regularizations of microscopic solutions, which correspond to trajectories of particles) or ``soliton-like'' (evolution of well-localized  ``bumps'') solutions and are analogues of multisoliton solutions of the Korteweg--de Vries equation. 

The rest of the paper is organized as follows. In Section~\ref{SecBE}, we remind the Boltzmann--Enskog equation. The notion of microscopic solution is reminded in Section~\ref{SecMicroSol}. In Section~\ref{SecSol}, we construct a regularization for microscopic solutions of the Boltzmann--Enskog equation for elastic hard spheres and so, give a rigorous sense to them. In Section~\ref{SecInel}, we do the same for the case of inelastic hard spheres. Finally, in Section~\ref{SecGE} we remind the generalized Enskog equation and, in Section~\ref{SecGESol}, construct explicit soliton-like solutions for it.

\section{Boltzmann--Enskog equation for elastic hard spheres}\label{SecBE}

The   Boltzmann--Enskog kinetic equation describes the dynamics of a gas of hard spheres. For elastic hard spheres, it has the form
\begin{equation}\label{EqBE}
\frac{\partial f}{\partial t}=-v_1\frac{\partial f}{\partial r_1}+Q(f,f),
\end{equation}
\begin{equation*}
Q(f,f)(r_1,v_1,t)=a^2\int_{\Omega_{v_1}}
(v_{21},\sigma)[f(r_1,v'_1,t)f(r_1+a\sigma,v'_2,t)-f(r_1,v_1,t)f(r_1-a\sigma,v_2,t)]
d\sigma dv_2.
\end{equation*}
Here $f=f(r_1,v_1,t) \geq0$ is the density function of hard spheres, where $r_1$ is the position of the center of a sphere, $v_1$ is its velocity. To avoid the boundary effects, let us consider a gas in a three-dimensional torus $\mathbb T^3=\mathbb R^3/(\alpha\mathbb Z\times \beta\mathbb Z\times \gamma\mathbb Z)$, where $\alpha,\beta,\gamma>a$. The dynamics of hard spheres in a rectangular box with elastic reflection from the walls can be reduced to the dynamics on a torus. Then $r_1\in\mathbb T^3$,
$v_1\in\mathbb R^3$, $t\in\mathbb R$, $\frac{\partial}{\partial r_1}$ is the gradient in $r_1$, $a>0$ is a positive constant (the diameter of a sphere).

Further,
\begin{equation}\label{EqVprime}
\begin{aligned}
v'_1&=v_1+(v_{21},\sigma)\sigma,\\
v'_2&=v_2-(v_{21},\sigma)\sigma,
\end{aligned}
\end{equation}
$v_{21}=v_2-v_1$, $\sigma\in S^2$ ($S^2$ is a unit sphere in $\mathbb{R}^3$),
$(\cdot,\cdot)$ is a scalar product. $v_1$ and $v_2$ have the meaning of the velocities of two hard spheres just before their collision, $v'_1$ and $v'_2$ are the corresponding velocities just after the collision. In such interpretation, $\sigma$ is a unit vector directed from the center of the second sphere to the center of the first sphere. The expression $Q(f,f)$ is called the collision integral. The integration is performed over the region

$$\Omega_{v_1}=\{(\sigma,v_2):\,\sigma\in S^2, v_2\in\mathbb R^3, (v_{21},\sigma)\geq0\}\subset S^2\times\mathbb R^3.$$

This equation differs from the Boltzmann kinetic equation for hard spheres by the terms $\pm a\sigma$ in the arguments of $f$ in the collision integral. Thus, the Boltzmann equation assumes the size of spheres to be negligibly small (in comparison to the scale of spatial variation of $f$), while the Boltzmann--Enskog equation takes the size of spheres into account. Formally, the Boltzmann equation for hard spheres is obtained from the Boltzmann--Enskog equation by the Boltzmann--Grad limit: $a\to0$, $f$ is rescaled according to $a^{-2}f$. Namely, one can show that a function of the form $a^{-2}f(r,v,t)$ where $f$ is a solution of the Boltzmann--Enskog equation tends to a solution of the Boltzmann equation in the limit $a\to0$ \cite{ArkCerc89,ArkCerc}. For recent investigations of the Boltzmann--Enskog equation see \cite{HaNoh}.

\section{Microscopic solutions of the Boltzmann-Enskog equation}\label{SecMicroSol}

Let $(q^0_1,w^0_1,\ldots,q^0_N,w^0_N)\in\mathbb R^{6N}$ and  $|q_i^0-q_j^0|>a$ for all $i\neq j$. Also let a system of hard spheres with initial positions $q^0_1,\ldots,q^0_N$ and velocities  $w^0_1,\ldots,w^0_N$ suffer only double collisions, don't suffer so called grazing collisions (with $(v_{21},\sigma)=0$ in (\ref{EqVprime})), and don't suffer infinite number of collisions in a finite time. The measure of the  excluded phase points is zero \cite{CPG}.

Let $q_1(t),\ldots,q_N(t)$ and   $w_1(t),\ldots,w_N(t)$ be positions and velocities of $N$ hard spheres at time $t$ provided that their initial state is $(q^0_1,w^0_1,\ldots,q^0_N,w^0_N)$.

Consider the generalized function
\begin{equation}\label{Eqf}
f(r_1,v_1,t)=\sum_{i=1}^N\delta(r_1-q_i(t))\delta(v_1-w_i(t)).
\end{equation}
N.\,N.~Bogolyubov proved that it satisfies the Boltzmann--Enskog equation \cite{Bogol75,BogBog}. He
referred to solutions of this form as ``microscopic solutions'', because they correspond to 
trajectories of individual hard spheres. The Boltzmann--Enskog
kinetic equation is believed to give only an approximate aggregate
description of the gas dynamics in terms of a single-particle
distribution function. The Bogolyubov's result shows that the
Boltzmann--Enskog equation can describe the gas dynamics in terms of
trajectories of individual hard spheres as well.

Another surprising fact is that the microscopic solutions (for the case of elastic collisions) are
time-reversible, i.e., if we replace $t\to-t$ and $v\to-v$, then a
microscopic solution transforms into another microscopic solution. This is a consequence of the reversibility of $N$ hard spheres dynamics.
So, the Boltzmann--Enskog kinetic equation, which is known to
describe an irreversible dynamics and entropy production, has
time-reversible solutions as well.

However, Bogolyubov performed the proof just on the ``physical''
level of rigour. Namely, the formal substitution of function
(\ref{Eqf}) into the Boltzmann--Enskog equation yields products of
delta-functions in the collision integral,
which are not well-defined. The time derivative of this function in the left-hand side of the equation
also yields such products because of discontinuity of the velocities
$w_i(t)$.

A way to give a rigorous sense to the notion of microscopic solutions is to introduce a regularization, i.e., to consider a family of smooth delta-like functions of the form

$$f_\varepsilon(r_1,v_1,t)=\sum_{i=1}^N\delta_\varepsilon(r_1-q_i(t),v_1-w_i(t)),$$
$\varepsilon>0$, where $\delta_\varepsilon(r,v)\to\delta(r,v)$ as $\varepsilon\to0$, and to prove that the difference between left- and right-hand sides of (\ref{EqBE}) vanishes in some space of generalized functions as $\varepsilon\to0$ if we substitute $f_\varepsilon$ instead of $f$. This is done in \cite{Trush} under the assumption that the function $\delta_\varepsilon$ is compactly supported and bounded. But the proof is rather involved. Here we  consider a special and, in some sense, more natural type of regularization, which leads to a simpler and more elegant proof.

\section{Regularization of microscopic solutions}\label{SecSol}

Let $\delta^i_\varepsilon(r,v):\mathbb T^3\times\mathbb R^3\to\mathbb R$, be a family of  functions, $\varepsilon>0$, $i=1,2,\ldots,N$. Let the following properties be satisfied for all $\varepsilon>0$ and $i=1,2,\ldots,N$:

\begin{enumerate}[(i)]
\item $\delta^i_\varepsilon(r,v)$ is continuously differentiable,

\item $\displaystyle\int_{\mathbb T^3\times\mathbb R^3}\delta^i_\varepsilon(r,v)drdv=1$,

\item $\delta^i_\varepsilon(r,v)$ is compactly-supported. Moreover, for every neighbourhood of the origin, there exists $\varepsilon_0>0$ such that the supports of $\delta^i_\varepsilon$ lies in this neighbourhood for  all $\varepsilon<\varepsilon_0$ and $i=1,2,\ldots,N$.
\end{enumerate}

It follows from these properties that $\delta^i_\varepsilon(r,v)\to\delta(r)\delta(v)$ as $\varepsilon\to0$, where $\delta (\cdot)$ is the (three-dimensional) delta function.

Let, further, $(q^0_1,w^0_1,\ldots,q^0_N,w^0_N)\in(\mathbb T^3\times\mathbb R^3)^N$ and  $|q_i^0-q_j^0|>a$ for all $i\neq j$. As before, let a system of hard spheres with initial positions $q^0_1,\ldots,q^0_N$ and velocities  $w^0_1,\ldots,w^0_N$ suffer only double collisions, don't suffer so called grazing collisions, and don't suffer infinite number of collisions in a finite time.

Consider the function $f_{\varepsilon}(r_1,v_1,t)$, $\varepsilon>0$, defined as
\begin{equation}\label{Eqfeps0}
f_{\varepsilon}(r_1,v_1,0)\equiv f^0_\varepsilon(r_1,v_1)=\sum_{i=1}^N\delta^i_\varepsilon(r_1-q_i^0,v_1-w_i^0),
\end{equation}
\begin{multline}\label{Eqfeps}
f_{\varepsilon}(r_1,v_1,t)=f_\varepsilon(r_1-v_1(t-T_k),v_1,T_k)\\+
\frac1{N-1}\int_{B^+_{r_1,v_1,t-T_k}} f_\varepsilon(r_1-v_1t^*-v'_1(t-t^*-T_k)\sigma,v'_1,T_k)\\\times f_\varepsilon(r_2-v_2t^*-v'_2(t-t^*-T_k),v'_2,T_k) \,dr_2dv_2,
\end{multline}
for $t\in(T_{k},T_{k+1}]$ and $k=0,1,2,\ldots$. Here
$$B^+_{r_1,v_1,t}=\{(r_2,v_2):\,|r_2-r_1|\geq a,\:t^*(r_1,v_1,r_2,v_2)<t\},$$
$t^*=t^*(r_1,v_1,r_2,v_2)$ is the time to collision of two hard spheres in their backward evolution from the phase points $(r_1,v_1)$ and $(r_2,v_2)$ ($t^*=\infty$ if there is no collision of two hard spheres in their backward evolution in the absence of other spheres), $v'_1$ and $v'_2$ are defined by (\ref{EqVprime}) with
\begin{equation}\label{Eqsigma}
\sigma=\sigma(r_1,v_1,r_2,v_2)=r_2-r_1-(v_2-v_1)t^*(r_1,v_1,r_2,v_2)\in S^2\quad (\text{if $t^*(r_2-r_1,v_2-v_1)<\infty$)}.
\end{equation}
Further, $0=T_0<T_1<T_2<\ldots$, $T_k\to\infty$, are such that every hard sphere suffers no more than one collision on every interval $[T_k,T_{k+1}]$ provided that the initial state of $N$ hard spheres is  $(q^0_1,w^0_1,\ldots,q^0_N,w^0_N)$. Moreover, let there be no collisions at the instances $T_1,T_2,\ldots$.

Let us consider the space $\mathscr T=C^1_0(\mathbb T^3\times\mathbb R^3\times[0,\infty))$ of functions $\varphi(r,v,t)$ that are continuous, continuously differentiable in $r$ and $t$ and satisfy the property $\varphi(r,v,t)=0$ for all $(r,v)\in \mathbb T^3\times\mathbb R^3$ if $t>T$ ($T$ may be different for different functions). This is a space of test functions. The space of generalized functions $\mathscr T'$ is the space of linear continuous functionals over $\mathscr T$. 

\begin{theorem}\label{ThSol}
\begin{equation}\label{EqThSol}
\frac{\partial f_\varepsilon}{\partial t}+v_1\frac{\partial f_\varepsilon}{\partial r_1}+Q(f_\varepsilon,f_\varepsilon)\to0
\end{equation}
as $\varepsilon\to0$ in the space $\mathscr T'$.
\end{theorem}

\begin{proof}
Consider the first interval $[0,T_1]$ and the function
$$D^\varepsilon(r_1,v_1,\ldots,r_N,v_N,t)=\frac{1}{N!}S^{(N)}_{-t}\left\lbrace\prod_{i=1}^Nf^0_\varepsilon(r_i,v_i)\right\rbrace.$$
Here, $S^{(N)}_t$ is the operator of the shift along the phase space trajectory of $N$ hard spheres:
$$S^{(N)}_t \varphi(r_1,v_1,\ldots,r_N,v_N)=\varphi(R_1(t),V_1(t),\ldots,R_N(t),V_N(t)),$$
where $\varphi$ is an arbitrary function and $R_1(t),V_1(t),\ldots,R_N(t),V_N(t)$ are the positions and velocities of $N$ hard spheres  at time $t$ provided that their initial positions and velocities are $r_1,v_1,\ldots,r_N,v_N$. The operator $S^{(N)}_t$ is defined almost everywhere in $G_N$, where
\begin{equation}\label{EqGk}
G_k=\{(r_{N-k+1},v_{N-k+1},\ldots,r_N,v_N)\in(\mathbb T^3\times\mathbb R^3)^k: |r_i-r_j|\geq a, \:i,j=1,\ldots,N,\:i\neq j\}.
\end{equation}

$D^\varepsilon$ can be regarded as a probability density function for a system of $N$ hard spheres. Consider the single-particle and two-particle distribution functions: 
\begin{equation}\label{EqF12}
\begin{aligned}
F^\varepsilon_1(r_1,v_1,t)&= N\int_{G_{N-1}} D^\varepsilon(r_1,v_1,\ldots,r_N,v_N,t)\prod_{j=2}^N dr_jdv_j,\\
F^\varepsilon_2(r_1,v_1,r_2,v_2,t)&=N(N-1)\int_{G_{N-2}} D^\varepsilon(r_1,v_1,\ldots,r_N,v_N,t)\prod_{j=3}^N dr_jdv_j.\end{aligned}
\end{equation}
They satisfy the first equation of the BBGKY hierarchy for elastic hard spheres (initially derived in \cite{Cerc}, rigorously justified in \cite{Spohn-BBGKY,IP,PetrinaGer}, a derivation for the most general case is given in \cite{Simon}):
\begin{multline}\label{EqBEBeps}
\frac{\partial F^\varepsilon_1(r_1,v_1,t)}{\partial t}=-v_1\frac{\partial F^\varepsilon_1(r_1,v_1,t)}{\partial r_1}+a^2\int_{\Omega_{v_1}}
(v_{21},\sigma)[F^\varepsilon_2(r_1,v'_1,r_1+a\sigma,v'_2,t)-\\-F^\varepsilon_2(r_1,v_1,r_1-a\sigma,v_1,t)]
d\sigma dv_2.
\end{multline}

In Lemma~\ref{LemFactor}, we prove that

\begin{equation}\label{EqFactor}
F_2^\varepsilon(r_1,v_1,r_1-a\sigma,v_2,t)=F_2^\varepsilon(r_1-v_1t,v_1,r_1-a\sigma-v_2t,v_2,0)= f^0_\varepsilon(r_1-v_1t,v_1)f^0_\varepsilon(r_1-a\sigma-v_2t,v_2)
\end{equation}
for all $r_1,v_1,v_2,$ and $\sigma$ such that $(v_2-v_1,\sigma)\geq 0$, for all $t\in[0,T_1]$, and for sufficiently small $\varepsilon$. As a corollary, since $(v'_2-v'_1,-\sigma)=(v_2-v_1,\sigma)\geq0$, we have
\begin{equation*}
F_2^\varepsilon(r_1,v'_1,r_1+a\sigma,v'_2,t)=F_2^\varepsilon(r_1-v'_1t,v'_1,r_1+a\sigma-v'_2t,v'_2,0)=f^0_\varepsilon(r_1-v'_1t,v'_1)f^0_\varepsilon(r_1+a\sigma-v'_2t,v'_2)
\end{equation*}
as well.

If we substitute function (\ref{Eqfeps}) into the Boltzmann--Enskog equation, we will obtain (remind that we consider the time interval $[0,T_1]$)
\begin{multline*}
\frac{\partial f_\varepsilon(r_1,v_1,t)}{\partial t}=-v_1\frac{\partial f_\varepsilon(r_1,v_1,t)}{\partial r_1}+a^2\int_{\Omega_{v_1}}
(v_{21},\sigma)[f^0_\varepsilon(r_1-v'_1t,v'_1)f^0_\varepsilon(r_1+a\sigma-v'_2t,v'_2)-\\-f^0_\varepsilon(r_1-v_1t,v_1)f^0_\varepsilon(r_1-a\sigma-v_2t,v_2)]
d\sigma dv_2.
\end{multline*}
By view of (\ref{EqFactor}), the collision integral coincides with that of the first equation of the BBGKY hierarchy. This is the crucial observation. We can rewrite the last equation as
\begin{multline}\label{EqFf}
\frac{\partial f_\varepsilon(r_1,v_1,t)}{\partial t}=-v_1\frac{\partial f_\varepsilon(r_1,v_1,t)}{\partial r_1}+a^2\int_{\Omega_{v_1}}
(v_{21},\sigma)[F_2^\varepsilon(r_1,v'_1,r_1+a\sigma,v'_2,t)-\\-F_2^\varepsilon(r_1,v_1,r_1-a\sigma,v_2,t)]
d\sigma dv_2.
\end{multline}

Let us consider the subspace $\mathscr T_{T_1}\subset \mathscr T$ of test functions $\varphi(r,v,t)\in\mathscr T$ with additional condition $\varphi(r,v,t)=0$ for all $r$ and $v$ if $t>T_1$ and the corresponding space of generalized functions $\mathscr T'_{T_1}$. Let us show that (\ref{EqThSol}) is satisfied in $\mathscr T'_{T_1}$. Comparing (\ref{EqBEBeps}) and (\ref{EqFf}), we see that, for this aim, it suffices to show that
$$
\frac{\partial f_\varepsilon(r_1,v_1,t)}{\partial t}+v_1\frac{\partial f_\varepsilon(r_1,v_1,t)}{\partial r_1}-\frac{\partial F_1^\varepsilon(r_1,v_1,t)}{\partial t}-v_1\frac{\partial F_1^\varepsilon(r_1,v_1,t)}{\partial r_1}\to0\quad\text{in }\mathscr T',
$$
or
\begin{equation}\label{EqLimfphi}
\int_{\mathbb T^3\times\mathbb R^3\times[0,\infty)}[f_\varepsilon(r_1,v_1,t)-F_1^\varepsilon(r_1,v_1,t)]\left[\frac{\partial \varphi(r_1,v_1,t)}{\partial t}+v_1\frac{\partial \varphi(r_1,v_1,t)}{\partial r_1}\right]dr_1dv_1dt\to0
\end{equation}
 as $\varepsilon\to0$, for an arbitrary test function $\varphi\in\mathscr{T}_{T_1}$.

According to Lemma~\ref{LemF1}, the function $F_1^\varepsilon(r_1,v_1,t)$ has the form
\begin{equation}\label{EqF1explicit}
\begin{split}
F_1^{\varepsilon}(r_1,v_1,t)&=\frac1{N-1} f^0_\varepsilon(r_1-v_1t,v_1)
\int_{B^-_{r_1,v_1,t}}f_\varepsilon^0(r_2-v_2t,v_2)\,dr_2dv_2\\&+
\frac1{N-1}\int_{B^+_{r_1,v_1,t}} f^0_\varepsilon(r_1-v_1t^*-v'_1(t-t^*)\sigma,v'_1)\,f^0_\varepsilon(r_2-v_2t^*-v'_2(t-t^*),v'_2) \,dr_2dv_2,
\end{split}
\end{equation}
where
$$B^-_{r_1,v_1,t}=\{(r_2,v_2):\,|r_2-r_1|\geq a,\:t^*(r_1,v_1,r_2,v_2)>t\}.$$
So, it differs from $f_\varepsilon(r_1,v_1,t)$ only by the factor 
$\frac1{N-1}\int_{B^-_{r_1,v_1,t}}f_\varepsilon^0(r_2-v_2t,v_2)\,dr_2dv_2$ in the first term. The region of integration $B^-_{r_1,v_1,t}$ is the set of states $(r_2,v_2)$ of the second sphere that do not lead to a collision with the first sphere in time $t$ of the backward evolution given the state of the first sphere is $(r_1,v_1)$. Let us express the first term of the right-hand side of (\ref{EqF1explicit}) as a sum:
\begin{multline}\label{EqFirstTerm}
\frac1{N-1} f^0_\varepsilon(r_1-v_1t,v_1)
\int_{B^-_{r_1,v_1,t}}f_\varepsilon^0(r_2-v_2t,v_2)\,dr_2dv_2\\=
\frac1{N-1}\sum_{i=1}^N \delta^{(i)}_\varepsilon(r_1-v_1t-q_i^0,v_1-w_i^0)
\int_{B^-_{r_1,v_1,t}}f_\varepsilon^0(r_2-v_2t,v_2)\,dr_2dv_2
\end{multline}
and consider the first term in this sum
$$\frac1{N-1} \delta^{(i)}_\varepsilon(r_1-v_1t-q_i^0,v_1-w_i^0)
\int_{B^-_{r_1,v_1,t}}f_\varepsilon^0(r_2-v_2t,v_2)\,dr_2dv_2.$$

If  sphere 1 suffers no collisions on the interval $[0,T_1]$, then 
\begin{equation}
\frac1{N-1}\int_{B^-_{r_1,v_1,t}}f_\varepsilon^0(r_2-v_2t,v_2)\,dr_2dv_2=
\frac1{N-1}\int_{(\mathbb T^3\backslash B_a(r_1))\times\mathbb R^3}f_\varepsilon^0(r_2-v_2t,v_2)\,dr_2dv_2=1,
\end{equation}
where $B_a(r_1)\subset\mathbb R^3$ is the ball with the center $r_1$ and the radius $a$.

Now let us assume that  sphere 1 collides with another sphere (without loss of generality let it have the number 2) and the moment of the collision is $t_1\in (0,T_1)$. Then for every $\Delta t>0$, there exist $\varepsilon>0$ such that all moments of collisions of spheres with initial phase points $(r_i,v_i)\in\supp\delta^{(i)}_\varepsilon(r_i-q_i^0,v_i-w_i^0)$, $i=1,2$, lies in the interval $(t_1-\Delta t,t_1+\Delta t)$. This means that
$$\frac1{N-1}\int_{B^-_{r_1,v_1,t}}f_\varepsilon^0(r_2-v_2t,v_2)\,dr_2dv_2=1$$ if $t<t_1-\Delta t$ (all phase points $(r_2,v_2)$ do not lead to a collision in the backward evolution on time $t$) and $$\delta^{(i)}_\varepsilon(r_1-v_1t-q_i^0,v_1-w_i^0)=0$$ if $t>t_1+\Delta t$ (all phase points $(r_2,v_2)$ lead to a collision, hence, there is no way for the first sphere to be in the phase point $(r_1,v_1)$ at the moment $t$ by free propagation from $(r_1-v_1t,v_1)$\footnote{Note that we use the absence of grazing collisions here. Otherwise, there are phase points $(r_1,v_1)$ such that the first sphere can get to them by two ways: by both free propagation and a collision with the second sphere (from different initial states of the two spheres within $\supp\delta^{(i)}_\varepsilon(r_i-q_i^0,v_i-w_i^0)$, $i=1,2$). Then $\delta^{(i)}_\varepsilon(r_1-v_1t-q_i^0,v_1-w_i^0)$ can be non-zero even if all phase points $(r_2,v_2)$ lead to a collision. But if the collision of the hard spheres with the initial states $(q_i^0,w_i^0)$, $i=1,2$, is non-grazing, then the two sets $$\{v_1|\,\exists r_1:\, (r_1,v_1)\in\supp\delta^{(1)}(r_1-q_1^0,v_1-w_1^0)\}\quad\text{and}$$ $$\{v'_1|\,\exists r_1,r_2,v_2:\,(r_1,v_1)\in\supp\delta^{(i)}(r_1-q_i^0,v_1-w_i^0),\,i=1,2\}$$ ($v'_1$ is the velocity of the first sphere after the collision of the spheres with the initial states $(r_i,v_i)$, $i=1,2$) do not intersect for sufficiently small $\varepsilon$, hence, there is at most one way to get to every phase point $(r_1,v_1)$: either by free propagation or by a collision (or neither).}). So,
$$\frac1{N-1} \delta^{(1)}_\varepsilon(r_1-v_1t-q_1^0,v_1-w_1^0)
\int_{B^-_{r_1,v_1,t}}f_\varepsilon^0(r_2-v_2t,v_2)\,dr_2dv_2=
\delta^{(1)}_\varepsilon(r_1-v_1t-q_1^0,v_1-w_1^0)$$
if $t<t_1-\Delta t$ or $t>t_1+\Delta t$.

Other terms in the sum of the right-hand side of (\ref{EqFirstTerm}) are analysed in the same way. Thus, the functions $f_\varepsilon$ and $F_1^\varepsilon$ coincide on $[0,T_1]$ apart from infinitesimal neighbourhoods of the moments of collisions. This proves limits (\ref{EqLimfphi}) and (\ref{EqThSol}) for $\mathscr T'_{T_1}$. As $\varepsilon\to0$, both $f_\varepsilon$ and $F_1^\varepsilon$ tend to the generalized function $\delta(r-q(t),v-w(t))$ defined as
$$(\delta(r-q(t),v-w(t)),\varphi(r,v,t))=\int_{0}^{\infty}\varphi(q(t),w(t),t)\,dt.$$

Since, by construction, there are no collisions at the moment $T_1$, the function $f_\varepsilon(r_1,v_1,T_1)$ has the same form as $f^0_\varepsilon(r_1,v_1)$ (according to Lemma~\ref{LemEvol}) and we can analyse the time interval $[T_1,T_2]$ in the same manner (regard $T_1$ as the initial time instant). That is, for the interval $[T_1,T_2]$, we set the function $D^\varepsilon$ to be
$$D^\varepsilon(r_1,v_1,\ldots,r_N,v_N,t)=\frac{1}{N!}S^{(N)}_{-(t-T_1)}\left\lbrace\prod_{i=1}^Nf_\varepsilon(r_i,v_i,T_1)\right\rbrace,$$
formulae (\ref{EqF12}), (\ref{EqBEBeps}), and (\ref{EqFf}) go trough, formula (\ref{EqFactor}) changes to
\begin{equation*}
\begin{split}
F_2^\varepsilon(r_1,v_1,r_1-a\sigma,v_2,t)&=F_2^\varepsilon(r_1-v_1(t-T_1),v_1,r_1-a\sigma-v_2(t-T_1),v_2,T_1)\\&= f_\varepsilon(r_1-v_1(t-T_1),v_1,T_1)f^0_\varepsilon(r_1-a\sigma-v_2(t-T_1),v_2,T_1).
\end{split}
\end{equation*}
Again, $F_1^\varepsilon(r_1,v_1,t)$ and $f^0_\varepsilon(r_1,v_1,t)$ coincide on $[T_1,T_2]$ apart from an infinitesimal time interval, which does not matter due to the integration with a test function. Hence, (\ref{EqThSol}) is valid on $\mathscr T'_{T_2}$.

We can continue this procedure up to an arbitrarily large time. For an arbitrary interval $[T_{k},T_{k+1}]$ we have
\begin{equation}\label{EqD}
D^\varepsilon(r_1,v_1,\ldots,r_N,v_N,t)=\frac{1}{N!}S^{(N)}_{-(t-T_k)}\left\lbrace\prod_{i=1}^Nf_\varepsilon(r_i,v_i,T_k)\right\rbrace.
\end{equation}
The theorem is proved.
\end{proof}

\begin{lemma}\label{LemFactor}
Equality (\ref{EqFactor}) is satisfied for all $r_1,v_1,v_2,$ and $\sigma$ such that $(v_2-v_1,\sigma)\geq 0$, for all $t\in[0,T_1]$ and for sufficiently small $\varepsilon$.
\end{lemma}
\begin{proof}

By construction, each hard sphere suffers no more than one collision on the interval $[0,T_1]$ provided that the initial state of $N$ hard spheres is $(q^0_1,w^0_1,\ldots,q^0_N,w^0_N)$. Let $\varepsilon$ be  small such that the same is true for almost every initial state $(r_1,v_1,\ldots,r_N,v_N)\in\supp D^\varepsilon(\cdot,0)$ (``almost'' -- due to the excluded zero measure initial states, see above). Such $\varepsilon$ exists since grazing collisions are prohibited, hence, the required set of initial states is open.

The condition $(v_2-v_1,\sigma)\geq 0$ means that two hard spheres are just before their collision. Since every hard sphere suffers no more than one collision, they moved free before the time $t$, so,
\begin{equation*}\begin{split}
&F^\varepsilon_2(r_1,v_1,r_1-a\sigma,v_2,t)\\&=\frac{1}{(N-2)!}\int_{G_{N-2}} S^{(N)}_{-t}\left\lbrace f^0_\varepsilon(r_1,v_1)f^0_\varepsilon(r_1-a\sigma,v_2)
\prod_{i=3}^Nf^0_\varepsilon(r_i,v_i)\right\rbrace\prod_{j=3}^N dr_jdv_j\\
&=\frac{1}{(N-2)!}\int_{G_{N-2}}  f^0_\varepsilon(r_1-v_1t,v_1)f^0_\varepsilon(r_1-a\sigma-v_2t,v_2)S^{(N-2)}_{-t}\left\lbrace\prod_{i=3}^Nf^0_\varepsilon(r_i,v_i)\right\rbrace\prod_{j=3}^N dr_jdv_j\\
&=f^0_\varepsilon(r_1-v_1t,v_1)f^0_\varepsilon(r_1-a\sigma-v_2t,v_2),
\end{split}\end{equation*}
q.e.d. Here, in the second equality, we have used the fact of free motion of the two hard spheres before time $t$ for sufficiently small $\varepsilon$.

\end{proof}

\begin{lemma}\label{LemF1}
Formula (\ref{EqF1explicit}) is satisfied for $t\in[0,T_1]$.
\end{lemma}
\begin{proof}
Since each hard sphere suffers no more than one collision for almost every initial state $(r_1,v_1,\ldots,r_N,v_N)\in\supp D^\varepsilon(\cdot,0)$ given $\varepsilon$ is sufficiently small (see the beginning of the proof of the previous lemma),
\begin{equation}
\begin{split}
F_1^{\varepsilon}(r_1,v_1,t)&=\frac1{N-1}\int_{(\mathbb T^3\backslash B_a(r_1))\times\mathbb R^3} S^{(2)}_{-t}F^\varepsilon_2(r_1,v_1,r_2,v_2,0)\,dr_2dv_2\\&=
\frac1{N-1}\int_{(\mathbb T^3\backslash B_a(r_1))\times\mathbb R^3} S^{(2)}_{-t}\{f_\varepsilon^0(r_1,v_1)f^0_\varepsilon(r_2,v_2)\}\,dr_2dv_2,
\end{split}
\end{equation}
where $B_a(r_1)\subset\mathbb R^3$ is a ball with center $r_1$ and radius $a$. An explicit expression of the operator $S^{(2)}_{-t}$ gives exactly (\ref{EqF1explicit}). 
\end{proof}

\begin{lemma}\label{LemEvol}
Let $t\in[0,T_1]$ is such that $|q_i(t)-q_j(t)|>a$ for all $i\neq j$. Then the function $f_\varepsilon(r_1,v_1,t)$ has the form
$$f_\varepsilon(r_1,v_1,t)=n^{-1}\sum_{i=1}^N\delta^{i,t}_\varepsilon(r_1-q_i(t),v_1-w_i(t))$$
where the functions $\delta^{i,t}_\varepsilon(r,v)$ obeys properties (i)---(iii) of the functions $\delta^i_\varepsilon(r,v)$ stated at the beginning of this section.
\end{lemma}

\begin{proof}
As we see in the main part of the proof of Theorem~\ref{ThSol} and in the proof of Lemma~\ref{LemF1}, $f_\varepsilon(r_1,v_1,t)$ coincides with the function $F_1(r_1,v_1,t)$ defined by (\ref{EqF12}) for sufficiently small $\varepsilon$ if there are no collisions at the moment $t$. So, we have
\begin{equation}\label{EqCalc}\begin{split}
f_\varepsilon(r_1,v_1,t)&=\frac{1}{(N-1)!}\int_{G_{N-1}} S^{(N)}_{-t}\left\lbrace\prod_{i=1}^Nf^0_\varepsilon(r_i,v_i)\right\rbrace\prod_{j=2}^N dr_jdv_j\\
&=\frac{1}{(N-1)!}\int_{G_{N-1}} \prod_{i=1}^Nf^0_\varepsilon(R_i(-t),V_i(-t))\prod_{j=2}^N dr_jdv_j\\
&=\frac{1}{(N-1)!}\int_{G_N\times G_{N-1}} \prod_{i=1}^N\{f^0_\varepsilon(p_i^0,u_i^0)\delta(p_i^0-R_i(-t),u_i^0-V_i(-t))\,dp_i^0du_i^0\}\prod_{j=2}^N dr_jdv_j\\
&=\frac{1}{(N-1)!}\int_{G_N\times G_{N-1}} |J|^{-1}\prod_{i=1}^N\{f^0_\varepsilon(p_i^0,u_i^0)\delta(r_i-p_i(t),v_i-u_i(t))\,dp_i^0du_i^0\}\prod_{j=2}^N dr_jdv_j\\
&=\frac{1}{(N-1)!}\int_{G_N\times G_{N-1}} \prod_{i=1}^N\lbrace f^0_\varepsilon(p_i^0,u_i^0)\delta(r_i-p_i(t),v_i-u_i(t))\,dp_i^0du_i^0\rbrace\prod_{j=2}^N dr_jdv_j.
\end{split}\end{equation}
Here 
\begin{equation}\label{EqRV}
\begin{aligned}
R_i(-t)&= R_i(-t,r_1,v_1,\ldots,r_N,v_N),\\V_i(-t)&= V_i(-t,r_1,v_1,\ldots,r_N,v_N),\\i&=1,2,\ldots,N,\end{aligned}
\end{equation} are the positions and velocities of $N$ hard spheres at time $-t$ provided that their initial positions and velocities are $r_1,v_1,\ldots,r_N,v_N$;
 
\begin{equation*}\begin{aligned}p_i(t)&=p_i(t,p_1^0,u_1^0,\ldots,p_N^0,u_N^0)\\ u_i(t)&=u_i(t,p_1^0,u_1^0,\ldots,p_N^0,u_N^0),\\i&=1,\ldots,N,\end{aligned}\end{equation*} are positions and velocities of $N$ hard spheres at time $t$ provided that their initial positions and velocities are $p_1^0,u_1^0,\ldots,p_N^0,u_N^0$;
$$J=\frac{\mathcal D(R_1(-t),V_1(-t),\ldots,R_N(-t),V_N(-t))}{\mathcal D(r_1,v_1,\ldots,r_N,v_N)}$$
is the Jacobian, which emerges due to the definition of a composition of a $6N$-dimensional delta function with smooth functions (\ref{EqRV})  of arguments $r_1,\ldots,v_N$. 

According to the Liouville's theorem, $|J|=1$, but we must prove the continuously differentiability of functions (\ref{EqRV}) with respect to initial positions and velocities $r_1,\ldots,v_N$. Their continuous differentiability is not
obvious, since the velocities of the spheres are discontinuous at
moments of collisions. Also, the operator $S^{(N)}_{-t}$ is not defined if $|r_i-r_j|<a$ for some $i\neq j$, hence, the derivatives in $J$ are not defined if $|r_i-r_j|=a$ for some $i\neq j$. Functions (\ref{EqRV}) are continuously differentiable in a point $(r_1,v_1,\ldots,r_N,v_N)$ whenever:
\begin{enumerate}[(i)]
\item $|r_i-r_j|>a$ for all $i\neq j$;
\item there is no collision at the instant $-t$, i.e., 
\begin{equation}\label{EqRij}
|R_i(-t,r_1,v_1,\ldots,r_N,v_N)-R_j(-t,r_1,v_1,\ldots,r_N,v_N)|> a
\end{equation}
for all $i\neq j$.
\end{enumerate}

The reasoning is as follows. 
\begin{enumerate}[(a)]
\item The positions of hard spheres are continuous with respect to the
initial state. Hence, for every $\delta_t>0$, there exists $\delta>0$
such that $|r_i-q_i(t)|<\delta_t$ for all $i$ whenever
$|R_i(-t)-q_i^0|<\delta$ and $|V_i(-t)-w_i^0|<\delta$ for all $i$.

\item For every $\delta>0$, there exists $\varepsilon_\delta>0$ such that the
expression $\prod_{i=1}^Nf^0_\varepsilon(R_i(-t),V_i(-t))$, $\varepsilon<\varepsilon_\delta$, is non-zero
only if $|R_i(-t)-q_i^0|<\delta$ and $|V_i(-t)-w_i^0|<\delta$ for all
$i$. 

\item By the condition of the lemma, $|q_i(t)-q_j(t)|>a$ for all $i\neq j$. 

\item From (a), (b), and (c): for
sufficiently small $\varepsilon$, we have $|r_i-r_j|>a$ for all $i\neq j$, otherwise $\prod_{i=1}^Nf^0_\varepsilon(R_i(-t),V_i(-t))=0$. We have obtained (i).
 
\item From (b) and from $|q_i^0-q_j^0|>a$ for all $i\neq j$: $|R_i(-t)-R_j(-t)|>a$ for all $i\neq j$ for sufficiently small $\varepsilon$, otherwise $\prod_{i=1}^Nf^0_\varepsilon(R_i(-t),V_i(-t))=0$. We have obtained (ii).
\end{enumerate}

Thus, for sufficiently small $\varepsilon$, we integrate over a region of initial states $r_1,v_1,\ldots,r_N,v_N$ in which  functions (\ref{EqRV}) are continuously differentiable.

Now, $f_\varepsilon(r_1,v_1,t)$ can be expressed as
\begin{equation*}\begin{split}
f_\varepsilon(r_1,v_1,t)&=\frac{1}{(N-1)!}\int_{G_N\times G_{N-1}} \prod_{i=1}^N\{f^0_\varepsilon(p_i^0,u_i^0)\delta(r_i-p_i(t),v_i-u_i(t))\,dp_i^0du_i^0\}\prod_{j=2}^N dr_jdv_j\\
&=\frac{1}{(N-1)!}\int_{G_N} \delta(r_1-p_1(t),v_1-u_1(t))\prod_{i=1}^Nf^0_\varepsilon(p_i^0,u_i^0)\,dp_i^0du_i^0\\
&=\frac{n^{-1}}{(N-1)!}\int_{G_N} \delta(r_1-p_1(t),v_1-u_1(t))\sum_{(j_1,\ldots,j_N)}\prod_{i=1}^N\delta^{j_i}_\varepsilon(p_i^0-q_{j_i}^0,u_i^0-w_{j_i}^0)\,dp_i^0du_i^0,
\end{split}
\end{equation*}
where $\delta(r,v)\equiv\delta(r)\delta(v)$, the last sum is performed over all permutations $(j_1,\ldots,j_N)$ of the indices $1,\ldots,N$.

For every permutation, let us renumber the pairs $(p_i^0,u_i^0)$ according to this permutation:
\begin{equation*}\begin{split}
f_\varepsilon(r_1,v_1,t)
&=\frac{n^{-1}}{(N-1)!}\int_{G_N} \sum_{(j_1,\ldots,j_N)}\delta(r_1-p_{j_1}(t),v_1-u_{j_1}(t))\prod_{i=1}^N\delta^{j_i}_\varepsilon(p_{j_i}^0-q_{j_i}^0,u_{j_i}^0-w_{j_i}^0)\,dp_i^0du_i^0\\
&=\frac{n^{-1}}{(N-1)!}\int_{G_N} \sum_{(j_1,\ldots,j_N)}\delta(r_1-p_{j_1}(t),v_1-u_{j_1}(t))\prod_{i=1}^N\delta^i_\varepsilon(p_i^0-q_i^0,u_i^0-w_i^0)\,dp_i^0du_i^0\\
&=n^{-1}\int_{G_N} \sum_{i=1}^N\delta(r_1-p_{i}(t),v_1-u_{i}(t))\prod_{j=1}^N\delta^j_\varepsilon(p_j^0-q_j^0,u_j^0-w_j^0)\,dp_j^0du_j^0\\
&=n^{-1}\sum_{i=1}^N\delta^{i,t}_\varepsilon(r_1-q_i(t),v_1-w_i(t)),
\end{split}
\end{equation*}
where
$$\delta^{i,t}_\varepsilon(r_1-q_i(t),v_1-w_i(t))=\int_{G_N}\delta(r_1-p_{i}(t),v_1-u_{i}(t))\prod_{j=1}^N\delta^j_\varepsilon(p_j^0-q_j^0,u_j^0-w_j^0)\,dp_j^0du_j^0.$$
In the third equality we used the fact the integral depends on $j_1$ and does not depend on $j_2,\ldots,j_N$. The number of permutations with a fixed $j_1$ is $(N-1)!$.

Now the required properties of the functions $\delta^{i,t}_\varepsilon$ follow from the corresponding properties of the functions  $\delta^i_\varepsilon$ and the last expression.
\end{proof}

\begin{remark}\label{RemMild}
If we define the space of test functions as $\mathscr C=C_0(\mathbb T^3\times\mathbb R^3\times[0,\infty))$ of continuous (not necessarily differentiable in $r$ and $t$) functions $\varphi(r,v,t)$ such that $\varphi(r,v,t)=0$ if $t>T$ (again, $T$ may be different for different functions), then the Boltzmann--Enskog equation is asymptotically satisfied in the mild sense:
\begin{equation*}
\lim_{\varepsilon\to0}\left\{f_{\varepsilon}(r_1,v_1,t)- f_{\varepsilon}^0(r_1-v_1t,v_1)-\int_0^t Q(f_{\varepsilon},f_{\varepsilon})(r_1-v_1(t-s),v_1,s)\,ds\right\}=0
\end{equation*}
in $\mathscr C'$. Note that a global existence theorem for the mild solution of the Boltzmann--Enskog equation is proved in \cite{ArkCerc89,ArkCerc}.
\end{remark}

\section{The case of inelastic hard spheres}\label{SecInel}

Analogous solutions can be constructed for the case of inelastic hard spheres as well. If a collision is inelastic, then formula (\ref{EqVprime}) is modified as
\begin{equation*}
v'_{1,2}=v_1\pm\frac{1+\mu(g)}{2}(v_{2}-v_1,\sigma)\sigma,
\end{equation*} 
where $\mu(g)\in(0,1]$ is the coefficient of restitution ($\mu\equiv 1$ for elastic collisions), $g=|(v_2-v_1,\sigma)|$. The Boltzmann--Enskog equation for inelastic hard spheres has the form \cite{Gran}

\begin{equation*}
\frac{\partial f}{\partial t}=-v_1\frac{\partial f}{\partial r_1}+Q(f,f),
\end{equation*}
\begin{multline*}
Q(f,f)(r_1,v_1,t)=na^2\int_{\Omega_{v_1}}
(v_{2}-v_1,\sigma)[\chi f(r_1,v''_1,t)f(r_1+a\sigma,v''_2,t)\\-f(r_1,v_1,t)f(r_1-a\sigma,v_2,t)]
d\sigma dv_2.
\end{multline*}
Here 
\begin{equation}\label{EqVpprime}
v''_{1,2}=v_1\pm\frac{1+\mu(g'')}{2\mu(g'')}(v_{2}-v_1,\sigma)\sigma
\end{equation} 
are the velocities before the ``inverse collision'' (i.e., they transforms into $v_{1,2}$  after the collision), 
\begin{equation}\label{EqChi}
\chi=\mu(g'')^{-1}\left|\frac{\mathcal D(v''_1,v''_2)}{\mathcal D(v_1,v_2)}\right|
\end{equation} 
($\chi\equiv1$ for elastic collisions), $g''=|(v''_2-v''_1,\sigma)|$, $\frac{\mathcal D(\cdot)}{\mathcal D(\cdot)}$ denotes the Jacobian.

The microscopic solutions are again defined by (\ref{Eqf}), but with the new definition of the time evolution $q_i(t), w_i(t),i=1,\ldots,N$. Let us give a rigorous sense to them by introduction of a similar regularization.

Let $\delta^i_\varepsilon(r,v)$, $i=1,\ldots,N$, $\varepsilon>0$, and $f^0_\varepsilon(r_1,v_1)$ be the same functions as in the previous section. Let $f_\varepsilon(r_1,v_1,t)$ now be defined by formula 

\begin{multline}\label{EqfepsIn}
f_{\varepsilon}(r_1,v_1,t)=f_\varepsilon(r_1-v_1(t-T_k),v_1,T_k)\\+
\frac1{N-1}\int_{B^+_{r_1,v_1,t-T_k}}\chi f_\varepsilon(r_1-v_1t^*-v'_1(t-t^*-T_k)\sigma,v'_1,T_k)\\\times f_\varepsilon(r_2-v_2t^*-v'_2(t-t^*-T_k),v'_2,T_k) \,dr_2dv_2,
\end{multline}
for $t\in(T_{k},T_{k+1}]$ and $k=0,1,2,\ldots$. Here $B^+_{r_1,v_1,t}$, $t^*=t^*(r_1,v_1,r_2,v_2)$, and $\sigma=\sigma(r_1,v_1,r_2,v_2)$  are defined as before. The spaces of test functions $\mathscr T$ and generalized functions $\mathscr T'$ are also defined as before.

The evolution operator $S_t$ is now defined as
\begin{multline*}
S^{(N)}_t \varphi(r_1,v_1,\ldots,r_N,v_N)=\varphi(R_1(t),V_1(t),\ldots,R_N(t),V_N(t))\\\times\left|\frac{\mathcal D(R_1(t),V_1(t),\ldots,R_N(t),V_N(t))}{\mathcal D(r_1,v_1,\ldots,r_N,v_N)}\right|,
\end{multline*}
where, again, $R_1(t),V_1(t),\ldots,R_N(t),V_N(t)$ are the positions and velocities of $N$ inelastic hard spheres at time $t$ provided that their initial positions and velocities are $r_1,v_1,\ldots,r_N,v_N$.  Note that the Jacobian is not defined if $|r_i-r_j|=a$ for some $i\neq j$, because $R_k(t)$ and $V_k(t)$, $k=1,\ldots,N$, are not defined for $|r_i-r_j|<a$ (i.e., in any neighbourhood of a point where $|r_i-r_j|=a$). So, in contrast to the case of elastic spheres, the operator $S^{(N)}_{t}$ is now defined only if $|r_i-r_j|>a$ for all $i\neq j$ (the inequality is now strict).

However, the Jacobian is a constant function of time on time intervals between collisions and has jumps at times of collisions, i.e., if $|R_i(t)-R_j(t)|=a$ for some $i\neq j$. So, it has a finite limit as $(r_1,\ldots,v_N)$ tends to the hypersurface 
$$\partial G_N=\{(r_1,v_1,\ldots,r_N,v_N):\,|r_i-r_j|\geq a \text{ for all } i\neq j,\, |r_k-r_l|=a \text{ for some } k\neq l\}$$
from $G_N\backslash\partial G_N$ (remind that the definition of $G_N$ is given by (\ref{EqGk})). By this reason, we can define the operator $S^{(N)}_t$ on this hypersurface by continuity and consider it to be defined on all $G_N$, as in the case of elastic collisions.

\begin{theorem}\label{ThSolIn}
$$
\frac{\partial f_\varepsilon}{\partial t}+v_1\frac{\partial f_\varepsilon}{\partial r_1}+Q(f_\varepsilon,f_\varepsilon)\to0
$$
as $\varepsilon\to0$ in the space $\mathscr T'$.
\end{theorem}
The proof is completely analogous to that of Theorem~\ref{ThSol} (Remark~\ref{RemMild} is also valid here). Namely, the functions $D^\varepsilon$, $F_1^\varepsilon$, and $F_2^\varepsilon$ are constructed in the same manner (with the new definition of $S^{(N)}_{-t}$). Single- and two-particle distribution functions $F_1^\varepsilon$ and $F_2^\varepsilon$ satisfy the first equation of  the BBGKY hierarchy for a system of inelastic hard spheres \cite{PG1,PG2}:
\begin{multline}\label{EqBEBIn}
\frac{\partial F_1^\varepsilon}{\partial t}=-v_1\frac{\partial F_1^\varepsilon}{\partial r_1}+na^2\int_{\Omega_{v_1}}
(v_{2}-v_1,\sigma)[\chi F_2^\varepsilon(r_1,v''_1,r_1+a\sigma,v''_2,t)\\-F_2^\varepsilon(r_1,v_1,r_1-a\sigma,v_2,t)]
d\sigma dv_2.
\end{multline}

Let us comment this. In \cite{PG1,PG2}, only a particular case $\mu=const$ is considered. However, the analysis can be easily generalized to $\mu$ as a function of $g=|(v_2-v_1,\sigma)|$. The main property to be satisfied for the validity of (\ref{EqBEBIn}) is the following boundary condition.

\begin{lemma}\label{LemJacob}
Let the function $D(r_1,v_1,\ldots,r_N,v_N,t)$ be defined as
$$D(r_1,v_1,\ldots,r_N,v_N,t)=S^{(N)}_{-t}D_0(r_1,v_1,\ldots,r_N,v_N),$$
where $D_0$ is a continuous function equal to zero in some neighbourhood of forbidden configurations (i.e., there exists $\delta>0$ such that $D_0=0$ whenever $|r_i-r_j|<a+\delta$ for some $i\neq j$). Then the boundary condition
\begin{multline}\label{EqBoundCond}
D(r_1,v_1,\ldots,r_i,v_i,\ldots,r_i-a\sigma,v_j\ldots,r_N,v_N,t)\\=
\chi D(r_1,v_1,\ldots,r_i,v''_i,\ldots,r_i-a\sigma,v''_j\ldots,r_N,v_N,t)\end{multline}
is satisfied for almost all $r_1,v_1,\ldots,r_N,v_N,$ and $\sigma$ such that $(v_2-v_1,\sigma)<0$, and for all $t$. Here $v''_i$, $v''_j$, and $\chi$ are defined by (\ref{EqVpprime}) and (\ref{EqChi}) correspondingly (with 1 and 2 replaced by $i$ and $j$).
\end{lemma}
\begin{proof}
By definition of $S^{(N)}_{-t}$,
\begin{multline*}
D(r_1,v_1,\ldots,r_N,v_N,t)=D_0(R_1(-t),V_1(-t),\ldots,R_N(-t),V_N(-t))\\\times\left|\frac{\mathcal D(R_1(-t),V_1(-t),\ldots,R_N(-t),V_N(-t))}{\mathcal D(r_1,v_1,\ldots,r_N,v_N)}\right|.
\end{multline*}

Phase trajectories are continuous with respect to time and initial phase point almost everywhere between the collisions. As we said above, the Jacobian is a constant function of time on time intervals between collisions and has jumps at times of collisions. A collision happens if $|R_k(-t)-R_l(-t)|=a$ for some $k\neq l$, but the function $D_0$ is zero in the neighbourhood of these hypersurfaces. So, the function $D$ is continuous almost everywhere in $G^N\times\mathbb R$.

It is sufficient to consider the case $|r_k-r_l|>a$ for all pairs $k\neq l$ except the pair $(i,j)$. If we prove the equality for this case, then, by continuity of $D$, the equality is valid for the general case as well.

As we noted above, if $|r_i-r_j|=a$, then the operator $S^{(N)}_{-t}$ is defined as a limit from $G_N\backslash\partial G_N$. So, the left- and right-hand sides of equality (\ref{EqBoundCond}) should be understood as limits
$$\lim_{\begin{smallmatrix}r_j\to r_i-a\sigma,\\
|r_i-r_j|>a\end{smallmatrix}}D(r_1,v_1,\ldots,r_i,v_i,\ldots,r_j,v_j\ldots,r_N,v_N,t),
$$
and
$$\lim_{\begin{smallmatrix}r_j\to r_i-a\sigma,\\
|r_i-r_j|>a\end{smallmatrix}}\chi D(r_1,v_1,\ldots,r_i,v''_i,\ldots,r_j,v''_j\ldots,r_N,v_N,t).$$
Now we use continuity in $t$:
\begin{equation*}\begin{split}
D(r_1,v_1,\ldots,r_N,v_N,t)&=
\lim_{\begin{smallmatrix}r_j\to r_i-a\sigma,\\
|r_i-r_j|>a\end{smallmatrix}}D(r_1,v_1,\ldots,r_N,v_N,t)\\&=
\lim_{\Delta t\to0}\lim_{\begin{smallmatrix}r_j\to r_i-a\sigma,\\
|r_i-r_j|>a\end{smallmatrix}}D(r_1,v_1,\ldots,r_N,v_N,t+\Delta t)\\&=
\lim_{\Delta t\to0}\lim_{\begin{smallmatrix}r_j\to r_i-a\sigma,\\
|r_i-r_j|>a\end{smallmatrix}}D(R_1(-\Delta t),V_1(-\Delta t),\ldots,R_N(-\Delta t),V_N(-\Delta t),t)\\&\times
\left|\frac{\mathcal D(R_1(-\Delta t),V_1(-\Delta t),\ldots,R_N(-\Delta t),V_N(-\Delta t))}{\mathcal D(r_1,v_1,\ldots,r_N,v_N)}\right|\\&=
D(r_1,v_1,\ldots,r_i,v''_i,\ldots,r_j\sigma,v''_j\ldots,r_N,v_N,t)\\&\times
\lim_{\Delta t\to0}\lim_{\begin{smallmatrix}r_j\to r_i-a\sigma,\\
|r_i-r_j|>a\end{smallmatrix}}\left|\frac{\mathcal D(R_1(-\Delta t),V_1(-\Delta t),\ldots,R_N(-\Delta t),V_N(-\Delta t))}{\mathcal D(r_1,v_1,\ldots,r_N,v_N)}\right|.
\end{split}\end{equation*}
Obviously, for sufficiently small $\Delta t$ and $|r_j-r_i|-a$, we have $R_k(-\Delta t)=r_k-v_k\Delta t$, $V_k(-\Delta t)=v_k$ for all $k\neq i,j$, $\lim\limits_{r_j\to r_i-a\sigma}V_i(-\Delta t)=v''_i$, $\lim\limits_{r_j\to r_i-a\sigma}V_j(-\Delta t)=v''_j$,
$$R_i(-\Delta t)=r_i-v_i t^*-v''_i(\Delta t-t^*),$$ 
$$R_j(-\Delta t)=r_j-v_j t^*-v''_j(\Delta t-t^*).$$ 
Here $t^*=t^*(r_1,r_2)\in(0,\Delta T)$ is the moment of the collision of the spheres $i$ and $j$ in the backward propagation, i.e., $|R_i(-t^*)-R_j(-t^*)|=a$.

Further,
\begin{equation*}
\begin{split}
\lim_{\Delta t\to0}\lim_{\begin{smallmatrix}r_j\to r_i-a\sigma,\\
|r_i-r_j|>a\end{smallmatrix}}&\left|\frac{\mathcal D(R_1(-\Delta t),V_1(-\Delta t),\ldots,R_N(-\Delta t),V_N(-\Delta t))}{\mathcal D(r_1,v_1,\ldots,r_N,v_N)}\right|\\
=\lim_{\Delta t\to0}\lim_{\begin{smallmatrix}r_j\to r_i-a\sigma,\\
|r_i-r_j|>a\end{smallmatrix}}&\left|\frac{\mathcal D(R_i(-\Delta t),v''_i,R_j(-\Delta t),v''_j)}
{\mathcal D(r_i,v_i,r_j,v_j)}\right|\\
=\lim_{\Delta t\to0}\lim_{\begin{smallmatrix}r_j\to r_i-a\sigma,\\
|r_i-r_j|>a\end{smallmatrix}}&\left|\frac{\mathcal D(R_i(-\Delta t),v_i,R_j(-\Delta t),v_j)}
{\mathcal D(r_i,v_i,r_j,v_j)}\right|\cdot\left|\frac{\mathcal D(R_i(-\Delta t),v''_i,R_j(-\Delta t),v''_j)}
{\mathcal D(R_i(-\Delta t),v_i,R_j(-\Delta t),v_j)}\right|.
\end{split}
\end{equation*}
Now,
\begin{equation}\label{EqJacobV}
\lim_{\Delta t\to0}\lim_{\begin{smallmatrix}r_j\to r_i-a\sigma,\\
|r_i-r_j|>a\end{smallmatrix}}\left|\frac{\mathcal D(R_i(-\Delta t),v''_i,R_j(-\Delta t),v''_j)}
{\mathcal D(R_i(-\Delta t),v_i,R_j(-\Delta t),v_j)}\right|=\left|\frac{\mathcal D(v''_i,v''_j)}{\mathcal D(v_i,v_j)}\right|.
\end{equation}
The parameter
$$\sigma=\frac{R_i(-t^*)-R_j(-t^*)}a$$ depends on both $R_{i,j}(-\Delta t)$ and $v_{i,j}$ (because $R_{i,j}(-t^*)=R_{i,j}(-\Delta t)+v''_{i,j}(\Delta t-t^*)$). However, by properties of Jacobians, we can consider $R_{i,j}(-\Delta t)$ to be constant in the right-hand side of (\ref{EqJacobV}). Also, in the limit, $\Delta t\to0$, $R_{i,j}(-\Delta t)\to R_{i,j}(-t^*)$, and $\sigma$ does not depend on $v$.

To finish the proof, we only need to prove that 
$$\lim_{\begin{smallmatrix}r_j\to r_i-a\sigma,\\
|r_i-r_j|>a\end{smallmatrix}}\left|\frac{\mathcal D(R_i(-\Delta t),R_j(-\Delta t))}
{\mathcal D(r_i,r_j)}\right|=\frac1\mu,$$
where we  consider $v_{i,j}$ to be constant in the calculation of the left-hand side. Also remind that $\mu=\mu(|(v''_j-v''_i,\sigma)|)$. We have

$$\left|\frac{\mathcal D(R_i(-\Delta t),R_j(-\Delta t))}
{\mathcal D(r_i,r_j)}\right|=\left|\frac{\mathcal D(R_i(-\Delta t)-R_j(-\Delta t),R_i(-\Delta t)+R_j(-\Delta t))}
{\mathcal D(r_i-r_j,r_i+r_j)}\right|=\left|\frac{\mathcal DR}
{\mathcal Dr}\right|,$$
where 
\begin{equation}\label{EqRr}
r=\frac{r_i-r_j}a,\quad R=\frac{R_i(-\Delta t)-R_j(-\Delta t)}a.
\end{equation} 
Here we used that, by the law of momentum conservation, 
$$R_i(-\Delta t)+R_j(-\Delta t)=r_i+r_j-(v_i+v_j)\Delta t,\quad\frac{\mathcal D(R_i(-\Delta t)+R_j(-\Delta t))}{\mathcal D(r_i+r_j)}=1.
$$
Variables (\ref{EqRr}) correspond to a reference frame attached to the $j$th hard sphere. Let us introduce the cylindrical coordinates $r=(\rho,\varphi,z)$ and choose $z$-axis along $v_i-v_j\equiv v$, i.e., $v=(0,0,|v|)$. Then $\sigma=(\rho,\varphi,\sqrt{1-\rho^2})$ (see Figure~1, $a=1$ in our rescaled variables), $(v_1-v_2,\sigma)=\sqrt{1-\rho^2}$,
$$v''_i-v''_j\equiv v''=v-\frac{1+\mu}\mu(v,\sigma)\sigma=(\frac{1+\mu}\mu|v|\rho\sqrt{1-\rho^2},\varphi+\pi,|v|-\frac{1+\mu}\mu|v|(1-\rho^2))\equiv(v''_\rho,\varphi+\pi,v''_z),$$
$$R=\sigma-v''(\Delta t-t^*)=(\rho+v''_\rho(\Delta t-t^*),\varphi,\sqrt{1-\rho^2}-v''_z(\Delta t-t^*))\equiv(R_\rho,\varphi,R_z),$$
$$t^*=\frac{z-\sqrt{1-\rho^2}}{|v|},$$
$\mu=\mu(\sqrt{1-\rho^2})$. We have
\begin{align*}
\lim_{\Delta t\to0}\lim_{\begin{smallmatrix}r_j\to r_i-a\sigma,\\
|r_i-r_j|>a\end{smallmatrix}}\frac{\partial R_\rho}{\partial\rho}&=1-v''_\rho\frac{\partial t^*}{\partial\rho},\\
\lim_{\Delta t\to0}\lim_{\begin{smallmatrix}r_j\to r_i-a\sigma,\\
|r_i-r_j|>a\end{smallmatrix}}\frac{\partial R_z}{\partial\rho}&=-\frac\rho{\sqrt{1-\rho^2}}+v''_z\frac{\partial t^*}{\partial\rho}
\end{align*}
(the terms $(\partial v''_{\rho,z}/\partial\rho)(\Delta t-t^*)$ disappear in the limit). Note also that $\frac{R_\rho}\rho$ (the ratio of Jacobians of transformations from Cartesian to the cylindrical coordinates in $R$ and $r$) tends to unity in this limit. $R_\rho$ and $R_z$ depend on $z$ only through $t^*$. So,
$$\lim_{\Delta t\to0}\lim_{\begin{smallmatrix}r_j\to r_i-a\sigma,\\
|r_i-r_j|>a\end{smallmatrix}}\frac{\mathcal DR}
{\mathcal Dr}=
\begin{vmatrix}
1-v''_\rho\frac{\partial t^*}{\partial\rho} & 
-v''_\rho\frac{\partial t^*}{\partial z}\\
-\frac\rho{\sqrt{1-\rho^2}}+v''_z\frac{\partial t^*}{\partial\rho} &
v''_z\frac{\partial t^*}{\partial z}
\end{vmatrix}=v''_z\frac{\partial t^*}{\partial z}-v''_\rho\frac{\partial t^*}{\partial z}\frac\rho{\sqrt{1-\rho^2}}=-\frac1\mu.
$$
The lemma is proved.
\end{proof}

\begin{figure}\centering
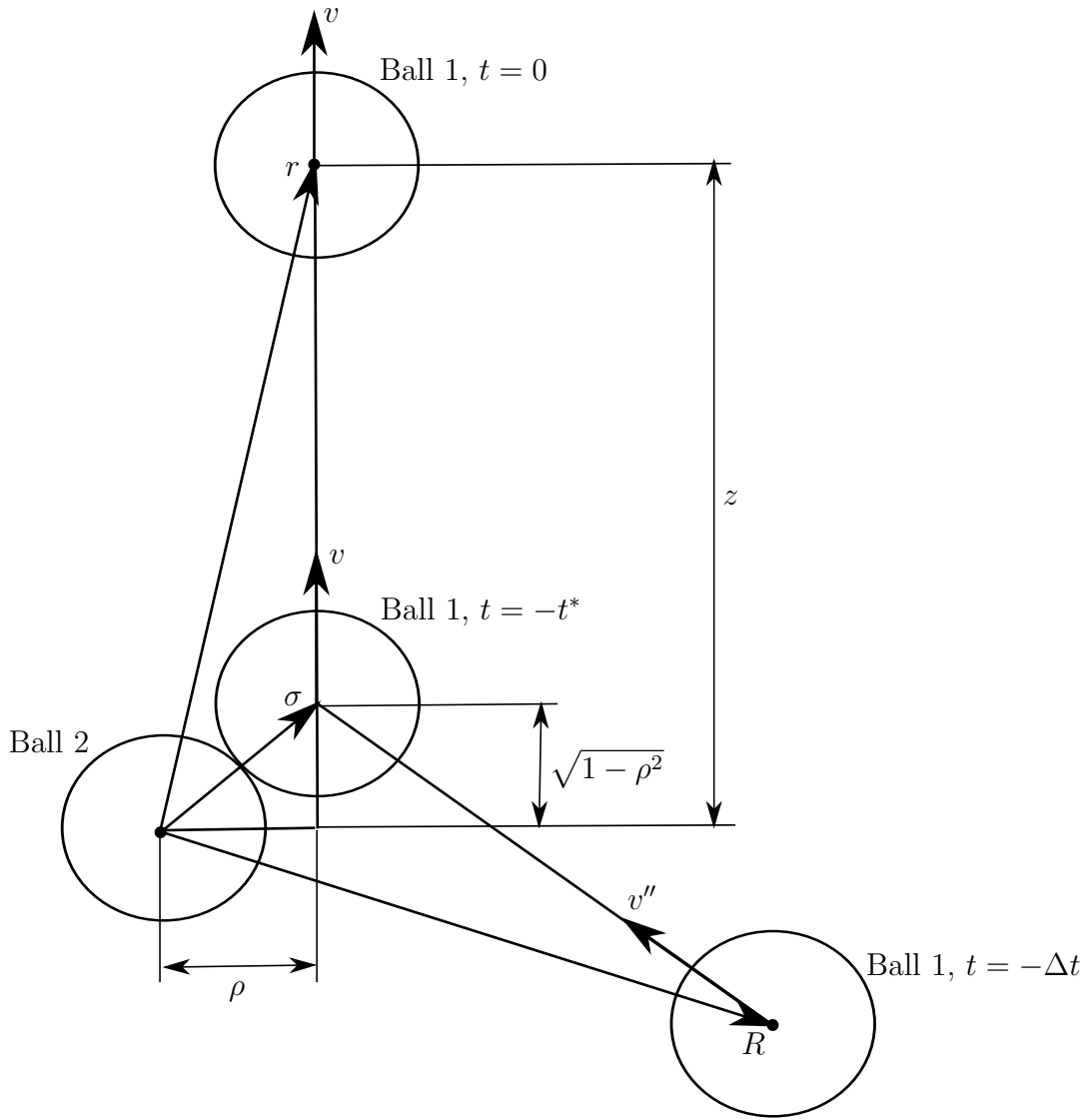
\caption{To the proof of Lemma~\ref{LemJacob}. Cylindrical coordinates $(\rho,\varphi,z)$, the plane $\varphi=const$. The frame of reference is attached to the second ball. The first ball is showed in three different time instants.}
\end{figure}

Again, $f_\varepsilon(r_1,v_1,t)$ and $F_1^\varepsilon(r_1,v_1,t)$ coincide apart from infinitesimal neighbourhoods of the moments of collisions. The counterparts of Lemmas~\ref{LemFactor}, \ref{LemF1}, and \ref{LemEvol} are proved in completely the same way as for the case of elastic hard spheres. In particular, calculations in (\ref{EqCalc}) are performed in such a way (the notations are the same):
\begin{equation*}\begin{split}
f_\varepsilon(r_1,v_1,t)&=\frac{1}{(N-1)!}\int_{G_{N-1}} S^{(N)}_{-t}\left\lbrace\prod_{i=1}^Nf^0_\varepsilon(r_i,v_i)\right\rbrace\prod_{j=2}^N dr_jdv_j\\
&=\frac{1}{(N-1)!}\int_{G_{N-1}} |J|\prod_{i=1}^Nf^0_\varepsilon(R_i(-t),V_i(-t))\prod_{j=2}^N dr_jdv_j\\
&=\frac{1}{(N-1)!}\int_{G_N\times G_{N-1}} |J|\prod_{i=1}^N\{f^0_\varepsilon(p_i^0,u_i^0)\delta(p_i^0-R_i(-t),u_i^0-V_i(-t))\,dp_i^0du_i^0\}\prod_{j=2}^N dr_jdv_j\\
&=\frac{1}{(N-1)!}\int_{G_N\times G_{N-1}}|J|\cdot|J|^{-1}\prod_{i=1}^N\{f^0_\varepsilon(p_i^0,u_i^0)\delta(r_i-p_i(t),v_i-u_i(t))\,dp_i^0du_i^0\}\prod_{j=2}^N dr_jdv_j\\
&=\frac{1}{(N-1)!}\int_{G_N\times G_{N-1}} \prod_{i=1}^N\lbrace f^0_\varepsilon(p_i^0,u_i^0)\delta(r_i-p_i(t),v_i-u_i(t))\,dp_i^0du_i^0\rbrace\prod_{j=2}^N dr_jdv_j.
\end{split}
\end{equation*}

\section{Generalized Enskog equation}\label{SecGE}

Thus, we see that smooth functions of form (\ref{Eqfeps}) (or (\ref{EqfepsIn}) for the case of inelastic hard spheres) satisfy the Boltzmann--Enskog equation asymptotically as $\varepsilon\to0$. It is interesting that  similar construction leads to exact smooth solutions for the recently proposed generalized Enskog equation \cite{GerasGap}, which is more general than the Boltzmann--Enskog one.

Here we again consider the case of elastic hard spheres, the generalization to the inelastic case is straightforward. The generalized Enskog equation has the form

\begin{equation}\label{EqGE}
\frac{\partial F_1}{\partial t}=-v_1\frac{\partial F_1}{\partial r_1}+Q_{GE}(F_1,F_1),
\end{equation}
\begin{multline*}
Q_{GE}(F_1,F_1)(r_1,v_1,t)=na^2\int_{\Omega_{v_1}}
(v_{21},\sigma)[F_2(r_1,v'_1,r_1+a\sigma,v'_2|F_1(t))\\-F_2(r_1,v_1,r_1-a\sigma,v_2|F_1(t))]
d\sigma dv_2.
\end{multline*}
Here $F_1=F_1(r_1,v_1,t)$ (now we suppose that $r_1\in\mathbb R^3$) has the meaning of the single-particle density function, $F_2=F_2(r_1,v_1,r_2,v_2|F_1(t))$ is the two-particle density function. It depends on time only through $F_1$:
\begin{equation}\label{EqF2GEE}
F_2(r_1,v_1,r_2,v_2|F_1(t))=
\sum_{n=0}^\infty\frac1{n!}\mathfrak V_t^{(1+n)}(\{1,2\},3,\ldots,n+2)\prod_{i=1}^{n+2}F_1(r_i,v_i,t),
\end{equation}

\begin{equation*}
\begin{split}
&\mathfrak V^{(1+n)}_t(\{Y\},X\backslash Y)=n!\sum_{k=0}^n(-1)^k
\sum_{m_1=1}^n\ldots\sum_{m_k=1}^{n-m_1-\ldots-m_{k-1}}
\frac1{(n-m_1-\ldots-m_{k})!}\\&\times
\widehat{\mathfrak A}_t^{(1+n-m_1-\ldots-m_{k})}(\{Y\},s+1,s+2,\ldots,s+n-m_1-\ldots-m_{k})\\&\times
\prod_{j=1}^k\sum_{k^j_2=0}^{m_j}\ldots
\sum_{k^j_{n-m_1-\ldots-m_{j}+s}=0}^{k^j_{n-m_1-\ldots-m_j+s-1}}
\prod_{i_j=1}^{s+n-m_1-\ldots-m_j}
\frac1{(k^j_{n-m_1-\ldots-m_j+s+1-i_j}-k^j_{n-m_1-\ldots-m_j+s+2-i_j})!}\\&\times
\widehat{\mathfrak A}_t^{(1+k^j_{n-m_1-\ldots-m_j+s+1-i_j}-
k^j_{n-m_1-\ldots-m_j+s+2-i_j})}(i_j,s+n-m_1-\ldots-m_j+1\\&
+k^j_{s+n-m_1-\ldots-m_j+2-i_j},\ldots,s+n-m_1-\ldots-m_j+k^j_{s+n-m_1-\ldots-m_j+1-i_j}),
\end{split}
\end{equation*}
where we mean $k^j_1=m_j$, $k^j_{n-m_1-\ldots-m_{j+s+1}}=0$. Also,

$$\widehat{\mathfrak{A}}_t^{(1+n)}(\{Y\},X\backslash Y)=
\mathfrak{A}_{-t}^{(1+n)}(\{Y\},X\backslash Y)\mathfrak{I}_{s+n}\prod_{i=1}^{s+n}
\mathfrak A^{(1)}_t(i).$$

The operator $\mathfrak A_{t}^{(1+n)}$ is  defined as follows:
$$\mathfrak A_{t}^{(1+n)}(\{Y\},X\backslash Y)=\sum_{P:\,(\{Y\},X\backslash Y)=\bigcup_iX_i}(-1)^{|P|-1}(|P|-1)!\prod_{X_i\in P}S^{(|\theta(X_i)|)}_{t}(\theta(X_i)).$$
The notation $S^{(N)}_{t}(i_1,i_2,\ldots,i_k)$ means that the phase shift operator acts on the variables with the numbers $i_1,\ldots,i_k$ (a function may depend on larger number of pairs of variables $(r_i,v_i)$).  Now we treat $(S^{(k)}_t(i_1,i_2,\ldots,i_k)\varphi)(r_1,v_1,\ldots,r_N,v_N)=0$ whenever $|r_i-r_j|<a$ for some $i\neq j$. If $k=N$, we will still write $S^{(N)}_{-t}\varphi$, because, in this case $i_1,\ldots,i_N$ are defined unambiguously: $i_1=1,\ldots,i_N=N$.

Further, $\{Y\}$ is a set consisting of a single element $Y=\{1,2,\ldots,s\}$, i.e., $|\{Y\}|=1$, $X=\{1,2,\ldots,s+n\}$. The sum is performed over all partitions $P$ of the set $\{\{Y\},X\backslash Y\}=\{\{Y\},s+1,\ldots,s+n\}$ into $|P|$ non-empty mutually disjoint subsets $X_i\subset \{\{Y\},X\backslash Y\}$. The mapping $\theta$ is the declasterization mapping defined by the formulae
\begin{equation*}
\begin{split}
\theta(\{\{1,\ldots,s\},i_1,i_2,\ldots,i_k\})&=\{1,\ldots,s,i_1,i_2,\ldots,i_k\},\\
\theta(\{i_1,i_2,\ldots,i_k\})&=\{i_1,i_2,\ldots,i_k\}.
\end{split}
\end{equation*}

The operator $\mathfrak{I}_{s+n}$ acts on functions $f_{s+n}(r_1,v_1,\ldots,r_{s+n},v_{s+n})$ as follows:
$$\mathfrak{I}_{s+n}f_{s+n}=\mathcal X_{s+n}f_{s+n},$$
where 
$$\mathcal X_k=\mathcal X_k(r_1,\ldots,r_k)=
\begin{cases}
1,&|r_i-r_j|\geq a\text{ for all } i\neq j,\\
0,& \text{otherwise}.\end{cases}$$

The operator $\mathfrak A_{-t}^{(n+1)}$ is called ``$(n+1)$-th order cumulant of operators $S$'', the operator $\widehat{\mathfrak A}_t^{(n+1)}$ is called ``$(n+1)$-th order scattering cumulant of operators $S$''. The cumulants  of the lowest orders have the form
\begin{align*}
&\mathfrak A_{-t}^{(1)}(\{1,2\})=S^{(2)}_{-t}(1,2),\\
&\mathfrak A_{-t}^{(2)}(\{1,2\},3)=S^{(3)}_{-t}(1,2,3)-S^{(2)}_{-t}(1,2)S^{(1)}_{-t}(3).
\end{align*}

It can be shown \cite{GerasGap} that the operator $\mathfrak{V}_t^{(1)}(\{1,2\})$ is the identity operator, so, the term $n=0$ in (\ref{EqF2GEE}) corresponds to the Boltzmann--Enskog equation. Thus, equation (\ref{EqGE}) is a generalization of the Boltzmann--Enskog equation.

The two-particle distribution function $F_2$ can be expressed in terms of the initial function $F^0_1(r,v)\equiv F_1(r,v,0)$ as follows:
\begin{multline}\label{EqF2GE}
F_2(r_1,v_1,r_2,v_2|F_1(t))=S^{(2)}_{-t}\left\lbrace\prod_{i=1}^{2}F^0_1(r_i,v_i)\right\rbrace\\+\sum_{n=1}^\infty\frac1{n!}\int_{\mathbb R^{6n}}\mathfrak{A}_{-t}^{(1+n)}(\{1,2\},3,\ldots,n+2)\mathcal X_{n+2}\prod_{i=1}^{n+2}F^0_1(r_i,v_i)\,dr_{3}dv_3\ldots dr_{n+2}dv_{n+2}.
\end{multline}
A solution of (\ref{EqGE}) is given by
\begin{multline}\label{EqGESol}
F_1(r_1,v_1,t)=F^0_1(r_1-v_1t,v_1)\\+\sum_{n=1}^\infty\frac1{n!}\int_{\mathbb R^{6n}}
\mathfrak A_{-t}^{(1+n)}(1,2,\ldots,n+1)\mathcal X_{n+1}\prod_{i=1}^{n+1}F^0_1(r_i,v_i)\,dr_{2}dv_2\ldots dr_{n+1}dv_{n+1}.
\end{multline}

Note that equation (\ref{EqGE}) is derived in such a way that it satisfies the dynamics of infinite number of particles with the initial $s$-particle distribution functions
$$F^0_s(r_1,v_1,\ldots,r_s,v_s)=\prod_{i=1}^sF_1^0(r_i,v_i)\mathcal X_s.$$

\section{Explicit soliton-like solutions of the generalized Enskog equation}\label{SecGESol}

Let $\delta^i_\varepsilon(r,v):\mathbb R^6\to\mathbb R$, $i=1,\ldots,N$, and $f^0_\varepsilon(r,v)$, $\varepsilon>0$, be the same functions as in Section~\ref{SecSol}. Let, additionally, $\varepsilon$ is such that  $\delta^i_\varepsilon(r,v)=0$ if $|r|>a$ for all $i$. Now consider the function $f_\varepsilon(r_1,v_1,t)$ defined as
\begin{equation}\label{EqGEfeps}
f_\varepsilon(r_1,v_1,t)=\frac{1}{(N-1)!}\int_{\mathbb R^{6(N-1)}} S^{(N)}_{-t}\left\lbrace\prod_{i=1}^Nf^0_\varepsilon(r_i,v_i)\right\rbrace\prod_{j=2}^N dr_jdv_j
\end{equation}

\begin{theorem}\label{ThGESol}
The function $f_{\varepsilon}(r_1,v_1,t)$ is a solution of the  generalized Enskog equation (\ref{EqGE}) for $t>0$.
\end{theorem}

\begin{proof}
Under the stated conditions, the series in (\ref{EqGESol}) has a maximal term $n=N-1$, because there are no $N+1$ points from the support of $F_1^0(r,v)\equiv f_\varepsilon^0(r,v)$ with the distance equal or greater than $a$ from each other. So, (\ref{EqGESol}) has the form
\begin{multline*}
F_1^\varepsilon(r_1,v_1,t)=f^0_\varepsilon(r_1-v_1t,v_1)\\+\sum_{n=1}^{N-1}\frac1{n!}\int_{\mathbb R^{6n}}
\mathfrak A_{-t}^{(n+1)}(1,2,\ldots,n+1)\mathcal X_{n+1}\prod_{i=1}^{n+1}f^0_\varepsilon(r_i,v_i)\,dr_{2}dv_2\ldots dr_{n+1}dv_{n+1}.
\end{multline*}
But this coincides with (\ref{EqGEfeps}). Indeed, if we substitute here the expression for $\mathfrak A_{n+1}$, then the only surviving term is $S^{(N)}_{-t}\{\prod_{i=1}^{N}f^0_\varepsilon(r_i,v_i)\}$, all other terms cancel each other.
\end{proof}

Thus, we obtained a family of explicit smooth solutions for the generalized Enskog equation. In general, this equation describes the dynamics of an infinite number of particles. The crucial observation is that special single-particle distribution functions admits only a finite number of particles. In this case, the solution of this equation can be expressed in terms of the evolution operator of a finite number of particles.

The constructed solutions can be compared to  multisoliton
(or particle-like) solutions of the Korteweg--de Vries equation: it also describes the evolution of particle-like ``bumps''. The
difference is that our solutions maintain their particle-like
(localized) form only for a finite time interval. Our solutions can be referred to as ``particle-like'' (regularizations of microscopic solutions, which correspond to trajectories of particles) or ``soliton-like'' solutions.

Let us consider the simplest nontrivial case $N=2$. Then equation
(\ref{EqGE}) for $f_\varepsilon$ is reduced to

\begin{equation}\label{EqGE2}
\frac{\partial f_\varepsilon}{\partial t}=-v_1\frac{\partial f_\varepsilon}{\partial r_1}+Q^{(2)}_{GE}(f_\varepsilon,f_\varepsilon),
\end{equation}
\begin{multline*}
Q^{(2)}_{GE}(f_\varepsilon,f_\varepsilon)(r_1,v_1,t)=na^2\int_{\Omega_{v_1}}
(v_{21},\sigma)f_\varepsilon(r_1,v'_1,t)f_\varepsilon(r_1+a\sigma,v'_2,t)\zeta(r_1,v'_1,t)\zeta(r_1+a\sigma,v'_2,t)\\-f_\varepsilon(r_1,v_1,t)f_\varepsilon(r_1-a\sigma,v_2,t)\zeta(r_1,v_1,t)\zeta(r_1-a\sigma,v_2,t)]
d\sigma dv_2,
\end{multline*}
where
$$\zeta(r_1,v_1,t)=\left(\int_{B^-_{r_1,v_1,t}} f^0_\varepsilon(r_2-v_2t,v_2)\,dr_2dv_2\right)^{-1}.$$

Indeed, since we deal with the unbounded configuration space, the two particles collide at most once. Hence, an explicit expression of (\ref{EqGEfeps}) is 
\begin{multline}\label{EqF1explicit2}
f_{\varepsilon}(r_1,v_1,t)= f^0_\varepsilon(r_1-v_1t,v_1)
\int_{B^-_{r_1,v_1,t}}f_\varepsilon^0(r_2-v_2t,v_2)\,dr_2dv_2\\+
\int_{B^+_{r_1,v_1,t}} f^0_\varepsilon(r_1-v_1t^*-v'_1(t-t^*)\sigma,v'_1)\,f^0_\varepsilon(r_2-v_2t^*-v'_2(t-t^*),v'_2) \,dr_2dv_2
\end{multline}
(cf. (\ref{EqF1explicit}), the notations are the same). As wee just saw, in this particular case, formula (\ref{EqGESol}) is reduced to the dynamics of two particles and, hence, (\ref{EqGE}) is reduced to the first equation of the BBGKY hierarchy 

\begin{multline}\label{EqBEBeps2}
\frac{\partial F^\varepsilon_1(r_1,v_1,t)}{\partial t}=-v_1\frac{\partial F^\varepsilon_1(r_1,v_1,t)}{\partial r_1}+a^2\int_{\Omega_{v_1}}
(v_{21},\sigma)[F^\varepsilon_2(r_1,v'_1,r_1+a\sigma,v'_2,t)-\\-F^\varepsilon_2(r_1,v_1,r_1-a\sigma,v_1,t)]
d\sigma dv_2.
\end{multline}
with $F^\varepsilon_1(r,v,t)=f_\varepsilon(r,v,t)$ and $F^\varepsilon_2(r_1,v_1,r_2,v_2,t)=S^{(2)}_{-t}\{\prod_{i=1}^2 f^0_\varepsilon(r_i,v_i)\}$. Since the particles collide at most once, (\ref{EqFactor}) is satisfied:

\begin{equation*}
F_2^\varepsilon(r_1,v_1,r_1-a\sigma,v_2,t)=F_2^\varepsilon(r_1-v_1t,v_1,r_1-a\sigma-v_2t,v_2,0)= f^0_\varepsilon(r_1-v_1t,v_1)f^0_\varepsilon(r_1-a\sigma-v_2t,v_2)
\end{equation*}
for all $r_1,v_1,v_2,$ and $\sigma$ such that $(v_2-v_1,\sigma)\geq 0$, for all $t\in[0,T_1]$. Since we consider the moment before the collision, the second term in the right-hand side of (\ref{EqF1explicit2}) is zero, because its integral is taken over the states that lead to collision in the backward evolution (see more detailed explanations in the proof of Theorem~\ref{ThSol}). So,
\begin{equation}\label{EqFactor2}
\begin{split}
F_2^\varepsilon(r_1,v_1,r_1-a\sigma,v_2,t)&=f^0_\varepsilon(r_1-v_1t,v_1)f^0_\varepsilon(r_1-a\sigma-v_2t,v_2)\\&=f_\varepsilon(r_1-v_1t,v_1,t)\zeta(r_1,v_1,t)f_\varepsilon(r_1-a\sigma-v_2t,v_2,t)\zeta(r_1-a\sigma,v_2,t)
\end{split}
\end{equation}
Substitution of (\ref{EqFactor2}) to (\ref{EqBEBeps2}) gives (\ref{EqGE2}).

We see, that the corrections  to the Boltzmann--Enskog equation can be expressed by additional factors $\zeta$ (depending on the initial distribution $f^0$). As we see in Section~\ref{SecSol}, if $\varepsilon\to0$, then these factors matter only on infinitesimal time intervals. 

Thus, for this special form of initial conditions, the generalized Enskog equation and its solutions have especially simple forms (\ref{EqGE2}) and (\ref{EqF1explicit2}) correspondingly. The function $t^*(r_1,v_1,r_2,v_2)=t^*(r_2-r_1,v_2-v_1)$ can be explicitly expressed: it is the minimal solution of the quadratic equation $|r_2-r_1-(v_2-v_1)t|=a$ (or $\infty$ if a real solution does not exist or if there are negative solutions), i.e.,
\begin{equation*}
t^*(r,v)=\begin{cases}
\frac1{|v|}\left(\frac{(v,r)}{|v|}-
\sqrt{a^2-r^2+\frac{(v,r)^2}{v^2}}\right),&|r|\geq a,\,a^2-r^2+\frac{(v,r)^2}{v^2}\geq 0,\,(v,r)\geq0,\\
\infty,&\text{otherwise},
\end{cases}
\end{equation*}
where $r=r_2-r_1$, $v=v_2-v_1$.

Finally, let us note that the generalization of the results of this section to the inelastic hard spheres is straightforward: one need only to redefine the operator $S_{-t}^{(N)}$ again and to add a factor $\chi$ (see Section~\ref{SecInel}) to the gain term of the collision integrals.

\section{Conclusions}

We have given a rigorous sense to the microscopic solutions discovered by N.\,N.~Bogolyubov. They are time-reversible. So, the reversibility/irreversibility property of the Boltzmann--Enskog equation depends on the  considered class of solutions. If the considered solutions have the form of sums of delta-functions, then this equation is reversible. If the considered solutions belong to the class of continuously differentiable functions, then the Boltzmann-Enskog equation is irreversible. The qualitative differences in solutions of evolution equations depending on the considered functional space is analysed in \cite{Kozlov, KozTr, KozTr2}, it is crucial in so called functional mechanics \cite{1,2,3,4,5,6,7,8}. 

Also we have found exact solutions for recently proposed generalized Enskog equation. The construction of these solutions is similar to the construction of regularizations of microscopic solutions for the Boltzmann--Enskog equation. These solutions are expressed in terms of the evolution operator for a finite number of hard spheres. In particular, for the case of two hard spheres, these solutions as well as the generalized Enskog equation itself have especially simple forms. 

The constructed solutions of the generalized Boltzmann--Enskog equation can be referred to as ``soliton-like'' solutions and  are analogues of multisoliton solutions of the Korteweg--de Vries equation. A difference is that our ``solitons'' spread out over time.

Also it is worthwhile to mention the particle-like (microscopic,
singular) wave propagations in solutions of the Boltzmann equation
for hard spheres, which dominate the short-time behaviour of 
solutions (see \cite{LiuYuGreen,LiuYu}, a generalization to hard
potential model is given in \cite{LeeLiuYu}). On the contrary,
fluid-like waves reveal the long-time dissipative behaviour of the
solutions. An investigation of relations between the microscopic solutions of the Boltzmann--Enskog equation and these
two types of waves in solutions of the Boltzmann equation for hard
spheres can be of interest.

\section*{Acknowledgements}
The author is grateful for useful remarks and discussions to H.~van Beijeren, T.\,V.~Dudnikova, V.\,V.~Kozlov, S.\,V.~Kozyrev, A.\,I.~Mikhailov, T.~Monnai, A.\,N.~Pechen, E.\,V.~Radkevich, S.~Sasa, O.\,G.~Smolyanov, H.~Spohn,  H.~Tasaki, D.\,V.~Treshchev, V.\,V.~Vedenyapin, I.\,V.~Volovich, and E.\,I.~Zelenov.   This work was partially supported by the Russian Foundation for Basic Research (project 12-01-31273-mol-a).

\end{document}